\documentclass[conference]{IEEEtran}
\IEEEoverridecommandlockouts
\usepackage{amsmath}
\usepackage{amssymb}
\usepackage{amsfonts}
\usepackage{graphicx}
\usepackage{subcaption}
\usepackage[numbers, square, comma, compress]{natbib}
\usepackage{stmaryrd}
\usepackage[ruled,linesnumbered]{algorithm2e}

\SetCommentSty{mycommfont}

\newtheorem{lemma}{Lemma}

\newtheorem{prop}{Proposition}
\newtheorem{ex}{Example}
\newtheorem{defn}{Definition}
\newtheorem{rem}{Remark}

\usepackage{tikz}\usetikzlibrary{calc}
\usetikzlibrary{shapes.geometric}

\renewcommand{\IEEEQED}{\IEEEQEDopen}
\providecommand{\abs}[1]{\ensuremath{\left\lvert #1 \right\rvert}}
\providecommand{\norm}[1]{\ensuremath{\left\Vert #1 \right\Vert}}
\providecommand{\vv}[1]{\textquotedblleft #1\textquotedblright}

\newcommand{\X}{\mathsf{X}}
\newcommand{\Y}{\mathsf{Y}}
\newcommand{\M}{\mathsf{M}}
\newcommand{\Z}{\mathsf{Z}}
\newcommand{\U}{\mathsf{U}}

\DeclareMathOperator*{\rank}{rank}
\providecommand{\dsb}[2]{\llbracket #1; #2 \rrbracket}

\usepackage{algorithmic}
\usepackage{graphicx}
\usepackage{textcomp}
\usepackage{xcolor}
\def\BibTeX{{\rm B\kern-.05em{\sc i\kern-.025em b}\kern-.08em
    T\kern-.1667em\lower.7ex\hbox{E}\kern-.125emX}}
\begin{document}
\addtolength{\topmargin}{0.06in}

\title{Finite-blocklength performance of polar wiretap codes under a total variation secrecy constraint\\
\thanks{L. Luzzi was supported in part by 
CY Initiative of Excellence (ANR-16-IDEX-0008), the France 2030 ANR program \vv{PEPR Networks of the Future} (ANR-22-PEFT-0009) and by the Horizon Europe/JU SNS project ROBUST-6G (GA no. 101139068). 
}
}

 \author{
  \IEEEauthorblockN{Laura Luzzi\IEEEauthorrefmark{2}\IEEEauthorrefmark{3}, Valerio Bioglio\IEEEauthorrefmark{1}} 
      \IEEEauthorblockA{\IEEEauthorrefmark{2}ETIS (UMR 8051 CY Cergy Paris Université, ENSEA, CNRS), e-mail: laura.luzzi@ensea.fr}        \IEEEauthorblockA{\IEEEauthorrefmark{3}\'Equipe COSMIQ, Inria de Paris, e-mail: laura.luzzi@inria.fr}
   \IEEEauthorblockA{\IEEEauthorrefmark{1}Dipartimento di Informatica, Università di Torino, e-mail: valerio.bioglio@unito.it}
}

\maketitle

\begin{abstract}
We study the performance of polarizing codes over a  degraded symmetric wiretap channel under a total variation distance (TVD) secrecy constraint. We show that the leakage can be bounded by the sum of the TVDs of the bit-channels corresponding to the confidential and frozen bits. In the asymptotic regime, this gives a new criterion to design wiretap codes with vanishing TVD leakage. In finite blocklength, it allows us to compute lower bounds for the secrecy rate of different families of polarizing wiretap codes over a binary erasure wiretap channel. 
\end{abstract}

\begin{IEEEkeywords}
Physical Layer Security, Wiretap channel, Polar codes, Finite Blocklength, Information Leakage, Total Variation Distance
\end{IEEEkeywords}

\section{Introduction}
Wiretap coding allows confidential communication without secret keys in the presence of passive adversaries, as long as an asymmetry in the channel quality between the legitimate receiver and the adversary can be guaranteed \cite{Wyner}. Widely adopted metrics for the information leakage include the (normalized or unnormalized) mutual information (MI) and the TVD \cite{Bloch_Laneman}.
In the asymptotic setting where the blocklength tends to infinity, secrecy capacity-achieving coding schemes have been proposed, notably based on polar codes \cite{Mahdavifar_Vardy}, which are an attractive solution since they are already part of the 5G New Radio standard. 
%
However, for applications requiring short packets or low latency, it is not possible in general to guarantee a vanishing information leakage; for a target leakage $\delta$ and error probability $\epsilon$, the gap of the finite-blocklength secrecy rate from the secrecy capacity must be taken into account.

Building on the finite-length channel coding rates analysis in \cite{polyanskiy2010channel}, authors in \cite{yang2019wiretap} proved tight bounds on the optimal second-order coding rate over discrete memoryless and Gaussian wiretap channels under a total variation (TV) secrecy constraint. 
Recent works have investigated the finite-length performance of practical wiretap codes. In \cite{pfister2017quantifying}, 
the equivocation of very short wiretap codes is quantified over a binary erasure channel (BEC); \cite{harrison2020exact} proposed algebraic methods to compute exact equivocation in special cases. 
Authors in \cite{nooraiepour2020secure} consider randomized Reed–Muller codes for Gaussian wiretap channels; 
other works have used deep learning to design wiretap codes \cite{fritschek2020deep,besser2020wiretap,Rana_Chou}. However, these approaches are currently limited to very short blocklengths, due to the complexity of estimating the leakage.

Wiretap Polar codes were considered in \cite{Taleb_Benammar_2021}, which used  a modified SC decoder to compute a lower bound for the MI leakage; \cite{mahdavifar2024finite} 
provided upper and lower bounds for the MI leakage of wiretap schemes based on polar codes as a function of their finite-length scaling; \cite{shakiba2021finite} compared the performance of polar and Reed-Muller codes under a TVD leakage constraint 
when the main channel is error-free and the eavesdropper's channel is a BEC; 
\cite{lin2024achieving} extended this study to multi-kernel PAC codes.  

In this paper, we adopt the same viewpoint as \cite{shakiba2021finite}, and consider a TV secrecy constraint.
While \cite{shakiba2021finite} only considered the 
case where the main channel is error-free,
we extend the approach to the general degraded symmetric wiretap channel. 
We show that the TV information leakage for polar-like wiretap codes is upper bounded by the sum of the TVDs of the bit-channels corresponding to the confidential bits and the frozen bits. 
In the asymptotic regime, this allows to propose a variant of the strong secrecy scheme in \cite{Mahdavifar_Vardy} which is tailored to the TV metric.
%
For finite blocklengths, our bound provides a simple criterion to design wiretap codes for a given target leakage $\delta$ and error probability $\epsilon$ with  guaranteed secrecy rate. In the case when the main channel and eavesdropper's channel are BECs, 
we compute explicit 
lower bounds for the secrecy rates of different families of polarizing codes --- namely Polar codes, Reed-Muller (RM) codes \cite{abbe2020reed}, multi-kernel (MK) polar codes \cite{bioglio2020multi} and adjacent-bit swapped (ABS) codes \cite{li2022adjacent}. 

\section{Notation and preliminaries}
\paragraph*{Notation} All logarithms are in base $2$. We use column notation for vectors, and denote by $x^n$ the vector $(x_1,\ldots,x_n)^t$, and by $x_i^j$ the subvector $(x_i,x_{i+1},\ldots,x_j)^t$. Given $\mathcal{A} \subseteq \llbracket 1; n \rrbracket$, the notation $x_{\mathcal{A}}$ refers to the projection of $x^n$ on the coordinates in $\mathcal{A}$. 
We denote the TVD between two discrete
distributions by $\mathbb{V}(p,q)=\frac{1}{2}\norm{p-q}_1$. 

\subsection{Properties of the TVD of symmetric channels}
We recall Gallager's definition 
\cite[p. 94]{Gallager}:
\begin{defn}
 A discrete channel $W :\mathcal{X} \to \mathcal{Y}$ is \emph{symmetric} if
 there is a partition $\mathcal{Y}=\mathcal{Y}_1 \cup \cdots \cup \mathcal{Y}_r$ 
 such that:
 \begin{enumerate}
 \item $\forall x,\bar{x} \in \mathcal{X}$, there exists a permutation $\pi_{x\bar{x}}:\mathcal{Y} \to \mathcal{Y}$ such that $\forall k \in \dsb{1}{r}$, $\pi_{x\bar{x}}(\mathcal{Y}_k)=\mathcal{Y}_k$, and $\forall y \in \mathcal{Y}, \; W(y|x)=W(\pi_{x\bar{x}}(y)|\bar{x})$, 
 \item $\forall k \in \dsb{1}{r}$, $\forall y,\bar{y} \in \mathcal{Y}_k$, there exists a permutation $\pi_{y\bar{y}}: \mathcal{X} \to \mathcal{X}$ such that 
 $\forall x \in \mathcal{X}, \; W(y|x)=W(\bar{y}|\pi_{y\bar{y}}(x))$.
 \end{enumerate}
 \end{defn}


 \begin{defn}
Consider a symmetric channel $W: \mathcal{X} \to \mathcal{Y}$ with transition probability $p_{\Y|\X}$. Given an input distribution $p_{\X}$ over $\mathcal{X}$ and the corresponding output $p_{\Y}= p_{\Y|\X} \circ p_{\X}$, let
 \begin{align*}
 &T(W, p_{\X})\!=
 \!\!\sum_{x \in \mathcal{X}, y \in \mathcal{Y}}\!\! \abs{p_{\X\Y}(x,y)\!-\!p_{\X}(x)p_{\Y}(y)}\!=
 \!2\mathbb{V}(p_{\X\Y},p_{\X}p_{\Y}).
 \end{align*}
 The \emph{TVD of the channel $W$} is defined as
 $$T(W)=\max_{p_{\X}} T(W,p_{\X}).$$
 \end{defn}


The following result may be of independent interest:
\begin{lemma} \label{lemma_TVD_maximized}
For a symmetric channel $W: \mathcal{X} \to \mathcal{Y}$ with transition probability $p_{\Y|\X}$, 
$$T(W) =T(W,p_{\bar{\X}}),$$
where $p_{\bar{\X}}$ is the uniform input distribution over $\mathcal{X}$ and $p_{\bar{\Y}}= p_{\Y|\X} \circ p_{\bar{\X}}$ is the corresponding output distribution.
\end{lemma}
\begin{IEEEproof}[Sketch of proof]
Recall that the TVD $\mathbb{V}(p,q)$ corresponds to the  $f$-divergence $D_f(p,q)$ with $f(u)=\frac{1}{2}\abs{u-1}$. Thus,  $T(W, p_{\X})$ is the corresponding $f$-information \cite[p. 134-135]{Polyanskiy_Wu}. It follows from \cite[Lemma 3.2]{csiszar1972class} that $T(W, p_{\X})$ is a continuous and concave function of $p_{\X}$. Note that although $f$ is not strictly convex, the 
sufficiency part of the second statement of \cite[Theorem 3.2]{csiszar1972class} still holds. 
If $W$ is symmetric, $p_{\bar{\Y}}$ is constant on each set $\mathcal{Y}_k$, $k \in \dsb{1}{r}$, and  $\mathbb{V}(p_{\Y|\X=x},p_{\bar{\Y}})$ takes the same value $K$ for all $x \in \mathcal{X}$. This implies that the information radius $R_f(W)=\inf_{p_{\Y}} \max_{x} \mathbb{V}(p_{\Y|\X=x},p_{\Y}) \leq K$. But by \cite[Theorem 3.2]{csiszar1972class}, $R_f(W)=\max_{p_{\X}} T(W, p_{\X})=K$. \end{IEEEproof}

\begin{ex} For a binary erasure channel (BEC) $W$ with erasure probability $p$, $T(W)=1-p$. 
\end{ex} 

\begin{defn}
A channel $Q: \mathcal{X} \to \mathcal{Z}$ is \emph{degraded} with respect to a channel $W: \mathcal{X} \to \mathcal{Z}$ if there exists a channel $\mathcal{P}: \mathcal{Y} \to \mathcal{Z}$ such that $\forall x \in \mathcal{X}, \;\forall z \in \mathcal{Z}$,
$Q(z|x)=\sum_{y \in \mathcal{Y}} W(y|x) \mathcal{P}(z|y)$. In this case, we write $Q \preccurlyeq W$.
\end{defn}
A proof of the next property can be found in \cite[Lemma 2]{shakiba2021finite}. 
\begin{lemma} \label{lemma_TVD_degraded}
If $Q \preccurlyeq W$, then $\forall p_{\X}$, $T(Q, p_{\X}) \leq T(W, p_{\X})$. 
\end{lemma}

The quality of a binary-input channel $W: \{0,1\} \to \mathcal{Y}$ can be measured by its \emph{Bhattacharyya parameter} 
$$Z(W)=\sum_{y \in \mathcal{Y}} \sqrt{W(y|0)W(y|1)} \in [0,1].$$ 
The following inequalities were proved in \cite[Appendix A]{arikan2009}:
\begin{lemma} \label{Arikan_TVD_bound}
For any binary-input memoryless symmetric channel (BMS) $W$,
\begin{align*}
&C(W) \leq T(W) \leq \sqrt{1-Z(W)^2}.
\end{align*}
\end{lemma}

\subsection{Channel polarization}
Let $n=2^m$. A polar-like code of parameters $(n,k)$ is specified by a pair $(G_n,\mathcal{A})$, where $G_n$ is an invertible $n \times n$ binary matrix, and $\mathcal{A} \subseteq \dsb{1}{n}$ is a set of information bits.  
Given an input $\M^k \in \{0,1\}^k$, the encoder sets $\U_{\mathcal{A}}=\M^k$ and $\U_{\dsb{1}{n} \setminus \mathcal{A}}=0^{n-k}$, and transmits $\X^n=G_n \U^n$.
In the case of polar codes, $G_n=G_2^{\otimes m}$, where $G_2=\begin{pmatrix} 1 & 0 \\ 1 & 1 \end{pmatrix}$ \cite{arikan2009}. 

Given a BMS channel $W: \{0,1\} \to \mathcal{Y}$, $\forall i=1,\ldots,n$ we define the $i$-th \emph{bit-channel} $W_n^{(i)}: \{0,1\} \to \mathcal{Y}^n \times \{0,1\}^{i-1}$ with transition probabilities 
$$W_n^{(i)}(y^n,u_1,\ldots,u_{i-1}|u_i)=\sum_{u_{i+1},\ldots,u_n} \frac{1}{2^{n-1}} W^n(y^n|G_n u^n).$$
To simplify notation, we will write $W^{(i)}$ instead of $W_n^{(i)}$ when $n$ is fixed. 
We say that $G_n$ is \emph{polarizing} over $W$ if the capacities $C(W^{(i)})$
are close to either $0$ or $1$ for almost all $i \in \{1,2,...,n\}$ as $n \rightarrow \infty$ \cite{li2022adjacent}. 


\subsection{Polarization of TVDs}
 Channel polarization under a TV metric
was studied in \cite{alsan2014re}, which showed that 
for a BMS
channel $W$, 
given an i.i.d. process $\{B_i\}$ with Bernoulli(1/2) distribution, if we define $T_0=T(W)$, and $\forall m\geq 0$,
$T_{m+1}= T(W_m^-)$ if $B_m=1$, and $T_{m+1}= T(W_m^+)$ if $B_m=0$,
where $W^+$ and $W^-$ denote Arikan's transforms \cite[eq. (17)-(18)]{arikan2009},
then $T_m$ is a bounded supermartingale in the interval $[0,1]$ and converges almost surely to $\{0,1\}$ \cite[Proposition 5.2]{alsan2014re}. Namely, the TVDs of the bit-channels polarize. Moreover, $T(W^-)=T(W)^2$, $T(W^{-})+T(W^{+}) \leq 2T(W)$ and 
$\sum_{i=1}^n T(W_n^{(i)}) \leq n T(W)$. Furthermore, the following result holds \cite[Proposition 5.8]{alsan2014re}:
\begin{prop} 
For any BMS
channel $W$, $\forall \beta \leq 1/2$,
$\lim_{m \to \infty} \mathbb{P}\{T_m < 2^{-2^{m\beta}}\}=1-C(W)$.
\end{prop}

\section{Secrecy bounds under a TVD metric} \label{section_wiretap}

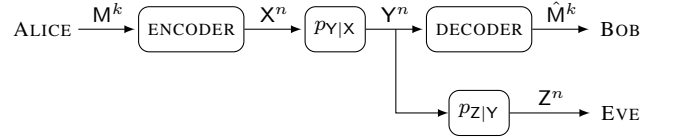
\begin{figure}[htb]
\begin{footnotesize}
\noindent\makebox[0.5\textwidth]{
\begin{tikzpicture}[
nodetype1/.style={
	rectangle,
	rounded corners,
	minimum width=8mm,
	minimum height=6mm,
	draw=black
},
nodetype2/.style={
	rectangle,
	rounded corners,
	minimum width=14mm,
	minimum height=6mm,
    draw=black
},
tip2/.style={-latex,shorten >=0.4mm}
]
\matrix[row sep=0.5cm, column sep=0.8cm, ampersand replacement=\&]{
\node (Alice) {\textsc{Alice}};  \& \node (encoder) [nodetype2]   {\textsc{encoder}}; \&
\node (W) [draw, nodetype1, text centered]  {$p_{\Y|\X}$}; \&
\node (decoder) [nodetype2] {\textsc{decoder}}; \& 
\node (Bob) {\textsc{Bob}};\\
\& (invisible) \& \&
\node (We) [draw, nodetype1, text centered]  {$p_{\Z|\Y}$}; \& 
\node (Eve) {\textsc{Eve}}; \\};
\draw[->] (Alice) edge[tip2] node [above] {$\M^k$} (encoder) ;
\draw[->] (encoder) edge[tip2] node [above] (X) {$\X^n$} (W) ;
\draw[->] (W) edge[tip2] node [above] (Y) {$\Y^n$} (decoder) ;
\draw[->] (decoder) edge[tip2] node [above] {$\hat{\M}^k$} (Bob) ;
\draw[->] (We) edge[tip2] node [above] {$\Z^n$} (Eve) ;
\draw[->,>=latex] (Y) |- node [anchor=east] {} (We);
\end{tikzpicture}}
\caption{The degraded wiretap channel.}
 \label{figure_degraded}
\end{footnotesize}
\end{figure}

We consider the 
degraded wiretap channel depicted in Figure \ref{figure_degraded}. We assume that Bob's channel $W_b : \mathcal{X} \to \mathcal{Y}$  and Eve's channel $W_e: \mathcal{X} \to \mathcal{Z}$, of transition probabilities $p_{\Y|\X}$ and $p_{\Z|\X}$ respectively, are both BMS channels.
The confidential message $\M^k$ is assumed to be uniform over $\{0,1\}^k$.

\subsection{Second order bounds for the secrecy rate}
We consider the average TVD $S(\M^k|\Z^n)$ between message $\M^k$ and Eve's observation $\Z^n$ as a measure of leakage \cite{yang2019wiretap}: 
\begin{equation} \label{average_TVD}
S(\M^k|\Z^n)=\mathbb{V}(p_{\M^k\Z^n},p_{\M^k}p_{\Z^n}).
\end{equation}
For this model, the secrecy capacity 
is 
$C_s=C(W_b)-C(W_e)$. In finite blocklength, \cite[Theorem 13]{yang2019wiretap} showed that for $\epsilon+\delta<1$, the maximal secrecy rate $R^*(n,\epsilon,\delta)$ for blocklength $n$, and average error probability $\epsilon$ under the secrecy constraint 
\begin{equation} \label{secrecy_constraint}
S(\M^k|\Z^n)\leq \delta
\end{equation}
is upper and lower bounded by
\begin{align} 
&\!\!\!\!\!R^*(n,\epsilon,\delta) \leq C_s -\sqrt{\frac{V_c}{n}}Q^{-1}\left(\epsilon+\delta\right)+\mathcal{O}\left(\frac{\log n}{n}\right),\label{YSP_Theorem13_upper} \\
&\!\!\!\!\!R^*(n,\epsilon,\delta)\! \geq \!C_s \!-\!\sqrt{\frac{V_b}{n}}Q^{-1}\!(\epsilon)\!-\!\sqrt{\frac{V_e}{n}}Q^{-1}\!(\delta)\!+\!\mathcal{O}\!\left(\!\frac{\log n}{n}\!\right)\!\! \label{YSP_Theorem13_lower}
\end{align}%
 where 
 $Q$ denotes the Q-function, $V_b$ and $V_e$ are the dispersions of Bob and Eve's channels, and
 \begin{multline}\label{V_c}
 \!\!\!\!\!\!V_c=\!\sum_{x \in \mathcal{X}} p_{\bar{\X}}(x) \!\Bigg(\!\sum_{y,z} p_{\Z\Y|\X}(z,y|x) \!\left(\!\log \frac{p_{\Z \Y|\X}(z,y|x)}{p_{\Z|\X}(z|x)p_{\Y|\Z}(y|z)}\!\right)^2 \\
 -\mathbb{D}\left(p_{\Z\Y|\X=x}||p_{\Y|\Z}p_{\Z|\X=x}\right)^2\Bigg),
 \end{multline}
provided that the uniform distribution $p_{\bar{\X}}$ is the unique distribution that maximizes $\mathbb{I}(\X;\Y|\Z)$
\cite{leung2003special}. 


\subsection{Wiretap coding scheme} \label{section_wiretap_polar}
We consider the general wiretap coding scheme in \cite{Mahdavifar_Vardy}. 
The set $\llbracket 1; n \rrbracket$ is partitioned into three disjoint subsets $\mathcal{A} \cup \mathcal{R} \cup \mathcal{B}$, where $\abs{\mathcal{A}}=k$ and $\abs{\mathcal{R}}=r$. Intuitively, $\mathcal{A}$ corresponds to the indices of bit-channels that are good for Bob but bad for Eve, $\mathcal{B}$ to bit-channels that are bad for both, and $\mathcal{R}$ to bit-channels that are good for both. 
Given the confidential message $\M^k \in \{0,1\}^k$ and a vector $\mathsf{V}^{r}$ of uniformly random bits, the wiretap encoder is defined by setting $\X^n=G_n \U^n$, where 
\begin{equation} \label{wiretap_encoder}
\U_{\mathcal{B}}=0^{n-k-r}, \; \U_{\mathcal{A}}=(\U_{i_1},\ldots,\U_{i_k})=\M^k, \; \U_{\mathcal{R}}=\mathsf{V}^{r}.
\end{equation}
The \emph{secrecy rate} of the wiretap code is $R_s=k/n$.

\subsection{Bound for the average error probability} 
Bob's block error probability under SC decoding is upper-bounded by the sum of Bhattacharyya parameters of Bob's unfrozen bit-channels in $\mathcal{G}=\mathcal{A} \cup \mathcal{R}$ \cite{Arikan_Telatar,fazeli2020binary}:
\begin{equation} \label{Bhattacharyya_bound}
P_e \leq \sum_{i \in \mathcal{A} \cup \mathcal{R}} Z(W_b^{(i)})
\end{equation}
\begin{rem}
Bob does not decode the bit-channels with indices greater than $i_{\max}(\mathcal{A})=\max\{j \in \mathcal{A}\}$, so that in (\ref{Bhattacharyya_bound}), the sum can be taken over $\mathcal{A} \cup \mathcal{R}'$, $\mathcal{R}'=\{i \in \mathcal{R}:\; i\leq i_{\max}(\mathcal{A})\}$. 
\end{rem}

\subsection{Bounds for the leakage} 
First, we prove a preliminary lemma (see Bound 1 in \cite{shakiba2021finite}): 
\begin{lemma} \label{Lemma_Mahdi}
If $\M^k$ is uniformly distributed in $\{0,1\}^k$, 
$$\mathbb{V}(p_{\M^k\Z^n},p_{\M^k}p_{\Z^n})\leq \sum_{i=1}^k \mathbb{V}(p_{\M^i\Z^n},\frac{1}{2}p_{\M^{i-1}\Z^n})$$
\end{lemma}
\begin{IEEEproof}
Letting $f_i(m^k,z^n)=\frac{1}{2^{k-i}} p_{\M^i\Z^n}(m^i,z^n)$ for $i=0,\ldots,k$, we can write
{\allowdisplaybreaks
\begin{align}
&\mathbb{V}(p_{\M^k\Z^n},p_{\M^k}p_{\Z^n}) \notag\\
&=\frac{1}{2} \sum_{m^k, z^n} \abs{f_k(m^k,z^n)-f_0(m^k,z^n)} \notag\\
&=\frac{1}{2} \sum_{m^k,z^n} \abs{ \sum_{i=1}^k (f_i(m^k,z^n)-f_{i-1}(m^k,z^n))} \notag\\
& \leq \frac{1}{2} \sum_{i=1}^k \sum_{m^k,z^n} \abs{f_i(m^k,z^n)-f_{i-1}(m^k,z^n)} \notag\\
&=\sum_{i=1}^k \mathbb{V}(p_{\M^i\Z^n},\frac{1}{2}p_{\M^{i-1}\Z^n}). \tag*{\IEEEQED}
\end{align}
}
\let\IEEEQED\relax
\end{IEEEproof}

\begin{lemma} \label{lemma_bounds}
Let $\Z^n$ be the output of the wiretap encoder with input $\U^n$ given by  (\ref{wiretap_encoder}) through the eavesdropper's channel $W_e$, and $\bar{\Z}^n$ the output of the polar encoder with uniform input $\bar{\U}^n$, and let $\mathcal{A} \cup \mathcal{B}=\{i_1,\ldots,i_{n-r}\}$. Then 
the following bounds hold for the leakage in TVD:
{\allowdisplaybreaks
\begin{align} 
S(\M^k|\Z^n) 
&\leq \frac{1}{2} \sum_{j=1}^{n-r} T(p_{\bar{\Z}^n\bar{\U}_{i_1}\bar{\U}_{i_2}\cdots\bar{\U}_{i_{j-1}}|\bar{\U}_{i_j}})
\tag*{(Bound 1)}\\
&\leq
\frac{1}{2}\sum_{i \in \mathcal{A}\cup \mathcal{B}} T(W_e^{(i)}). 
\tag*{(Bound 2)}
\end{align}
}
\end{lemma}
\begin{IEEEproof}
Let $\mathcal{A} \cup \mathcal{B}=\{i_1,\ldots,i_{n-r}\}$. We have
{\allowdisplaybreaks
\begin{align*} 
S(\M^k|\Z^n) 
&=\mathbb{V}(p_{\U_{\mathcal{A}}\Z^n},p_{\U_{\mathcal{A}}}p_{\Z^n}) \notag \\
&\stackrel{(a)}{=}\mathbb{V}(p_{\U_{\mathcal{A}\cup \mathcal{B}}\Z^n},p_{\U_{\mathcal{A}\cup \mathcal{B}}}p_{\Z^n}) \notag\\
&\stackrel{(b)}{\leq} \mathbb{V}(p_{\bar{\U}_{\mathcal{A}\cup \mathcal{B}}\bar{\Z}^n},p_{\bar{\U}_{\mathcal{A}\cup \mathcal{B}}}p_{\bar{\Z}^n}) \notag\\
&=\frac{1}{2}T(p_{\bar{\Z}^n|\bar{\U}_{\mathcal{A}\cup \mathcal{B}}}) \notag\\
&\stackrel{(c)}{\leq} \sum_{j=1}^{n-r} \mathbb{V}\left(p_{\bar{\U}_{i_1}\bar{\U}_{i_2}\cdots\bar{\U}_{i_j}\bar{\Z}^n},p_{\bar{\U}_{i_j}}p_{\bar{\U}_{i_1}\bar{\U}_{i_2}\cdots\bar{\U}_{i_{j-1}}\bar{Z}^n}\right) \notag\\
&= \frac{1}{2} \sum_{j=1}^{n-r} T(p_{\bar{\Z}^n\bar{\U}_{i_1}\bar{\U}_{i_2}\cdots\bar{\U}_{i_{j-1}}|\bar{\U}_{i_j}}), \notag 
\end{align*}
}%
where (a) holds since the bits $\U_{\mathcal{B}}$ in (\ref{wiretap_encoder}) are fixed to zero, (b) holds  by Lemma \ref{lemma_TVD_maximized} since the channel $Q_n(W_e,\mathcal{R})$ with transition probabilities $p_{\bar{\Z}^n|\bar{\U}_{\mathcal{A}\cup \mathcal{B}}}$ is symmetric provided that $W_e$ is symmetric \cite[Proposition 13]{Mahdavifar_Vardy}, and (c) by Lemma \ref{Lemma_Mahdi}.\\
Assuming now that $i_1<i_2 <\ldots < i_{n-r}$, we have
{\allowdisplaybreaks
\begin{align*}
\frac{1}{2} \sum_{j=1}^{n-r} T(p_{\bar{\Z}^n\bar{\U}_{i_1}\bar{\U}_{i_2}\cdots\bar{\U}_{i_{j-1}}|\bar{\U}_{i_j}}) 
&\stackrel{(d)}{\leq} \frac{1}{2} \sum_{i \in \mathcal{A}\cup \mathcal{B}} T(p_{\bar{\Z}^n\bar{\U}_1^{i-1}|\bar{\U}_i}) \notag \\
&=\frac{1}{2}\sum_{i \in \mathcal{A}\cup \mathcal{B}} T(W_e^{(i)}). 
\end{align*}
}%
where (d) holds since the channel $p_{\bar{\Z}^n\bar{\U}_{\llbracket 1; i-1 \rrbracket \cap (\mathcal{A}\cup \mathcal{B})}|\bar{\U}^i}$ is degraded with respect to the channel $p_{\bar{\Z}^n\bar{\U}_1^{i-1}|\bar{\U}^i}$.
\end{IEEEproof}

\begin{rem}
Bounds 1 and 2 in Lemma \ref{lemma_bounds} can be seen as the TV analogue of the upper bounds in \cite[eq. (77)]{Mahdavifar_Vardy} and \cite[Proposition 1]{mahdavifar2024finite} for MI. However, it does not seem possible to prove a lower bound for the TVD leakage analogously to \cite[Proposition 3]{mahdavifar2024finite} due to the fact that TV does not tensorize.
\end{rem}

\section{Wiretap code design}
\subsection{Wiretap code scheme with vanishing total variation leakage}
Using Lemma \ref{lemma_bounds}, we can adapt the wiretap scheme in \cite{Mahdavifar_Vardy} in order to guarantee vanishing TV leakage. 
\begin{prop} \label{prop_asymptotic}
For $\beta \in (0,1/2)$ and $c_1 2^{-n^{\beta}} < \delta_n =o(1/n)$ for some constant $c_1$, define 
\begin{align*}
&\mathcal{G}_n(W_b)=\{i \in \dsb{1}{n}: \; Z(W_b^{(i)})<2^{-n^\beta}\!/n\},\\
&\mathcal{P}_n(W_e)=\{i \in \dsb{1}{n}: \; T(W_e^{(i)})\leq\delta_n\},
\end{align*}
and let
\begin{align*}  
&\mathcal{R}=\dsb{1}{n} \setminus \mathcal{P}_n(W_e),\\
&\mathcal{A}=\mathcal{G}_n(W_b) \cap \mathcal{P}_n(W_e), \\
&\mathcal{B}=\mathcal{P}_n(W_e) \setminus \mathcal{G}_n(W_b).
\end{align*}
As $n \to\infty$, the wiretap scheme (\ref{wiretap_encoder}) with polarizing transform $G_n$ guarantees 
\begin{align*}
&\lim_{n \to \infty} S(\M^k|\Z^n) = 0,\\
&\lim_{n \to \infty} R_s =C(W_b)-C(W_e).
\end{align*}
\end{prop}
\begin{rem} The TV leakage constraint can be seen as intermediate between the weak secrecy and the strong secrecy constraint \cite[Proposition 1]{Bloch_Laneman}. Although the strong secrecy scheme in \cite{Mahdavifar_Vardy} directly guarantees vanishing TV leakage by Pinsker's inequality, and the two schemes have the same asymptotic secrecy rate, the definition of the set $\mathcal{P}_n(W_e)$ in Proposition \ref{prop_asymptotic} is less restrictive (by Lemma \ref{Arikan_TVD_bound}), and for finite $n$, the secrecy rate will in general be larger.
\end{rem}
 
\begin{IEEEproof}[Sketch of proof]
The argument is almost the same as for the strong secrecy scheme in \cite{Mahdavifar_Vardy}. 
The vanishing of the TV leakage follows from Bound 2 in Lemma \ref{lemma_bounds}, since
$$S(\M^k|\Z^n) \leq \sum_{i \in \mathcal{P}_n(W_e)} T(W_e^{(i)}) \leq n \delta_n \to 0.$$  
Note that 
$$R_s=\frac{\abs{\mathcal{A}}}{n} \geq \frac{\abs{\mathcal{P}_n(W_e)}}{n}-\frac{\abs{\dsb{1}{n} \setminus \mathcal{G}_n(W_b)}}{n}.$$
Recall that \cite[Theorem 1]{Mahdavifar_Vardy}
$$\lim_{n \to \infty} \abs{\dsb{1}{n} \setminus \mathcal{G}_n(W_b)}/n=1-C(W_b).$$
For $\beta < \alpha < 1/2$, define 
$$\mathcal{P}_n'=\{i \in \dsb{1}{n}:\; Z(W_e^{(i)})\geq 1-2^{-n^{\alpha}}\}.$$ By \cite[Corollary 19]{Mahdavifar_Vardy},
$\lim_{n\to \infty} \abs{\mathcal{P}_n'}/n=1-C(W_e)$.
Furthermore, by Lemma \ref{Arikan_TVD_bound}, $\forall i \in \mathcal{P}_n'$, 
$T(W_e^{(i)}) \leq 
2^{-(n^{\alpha}-1)/2}$.
Therefore, for large enough $n$, $\mathcal{P}_n'\subseteq \mathcal{P}_n(W_e)$ and $\lim_{n\to\infty} \abs{\mathcal{P}_n(W_e)}/n \geq \lim_{n\to \infty} \abs{\mathcal{P}_n'}/n=1-C(W_e)$ and 
\begin{equation} \label{R_s}
\lim_{n\to \infty} R_s \geq C(W_b)-C(W_e).
\end{equation}
Equality in (\ref{R_s}) is proven similarly to \cite[Proposition 22]{Mahdavifar_Vardy} by noting that $\mathcal{G}_n(W_b) \cap \mathcal{P}_n(W_b)=\varnothing$ since $T(W) \geq C(W) \geq \log\frac{2}{1+Z(W)}$ by \cite[Proposition 1]{arikan2009}.
\end{IEEEproof}
Note that similarly to the strong secrecy scheme in \cite{Mahdavifar_Vardy}, we can't guarantee vanishing error probability for Bob, due to the existence of a small set of unreliable and insecure bits.

\begin{algorithm}[bth]
\begin{small}
\SetAlgoLined
 \SetAlgoInsideSkip{smallskip}
 \SetKwInOut{Input}{Input}
 \SetKwInOut{Output}{Output}
\Input{ $W_e$, $W_b$, $n$, $\epsilon$, $\delta$}
\Output{$\mathcal{A} \subseteq \llbracket 1; n \rrbracket$, $\mathcal{G}=\mathcal{A} \cup \mathcal{R} \subseteq \llbracket 1; n \rrbracket$; lower bound $R_s$ for the secrecy rate}
\tcp{Find $\mathcal{G}=\mathcal{A} \cup \mathcal{R}$}
    Compute $Z(W_b^{(i)})$ for $i=1,\ldots,n$\;
    Sort $Z(W_b^{(i)})$ in increasing order and store the vector of sorting indices $\textrm{I}_b$\;
    $k_b=\max\left\{k\in\llbracket 1; n\rrbracket: \sum_{j=1}^k Z\left(W_b^{\textrm{I}_b(j)}\right)<\epsilon\right\}$ \;
    $\mathcal{G}=\{\textrm{I}_b(j) \;|\;j=1,\ldots,k_b\}$\;
    $\mathcal{B}=\{\textrm{I}_b(j) \;|\;j=k_b+1,\ldots,n\}$\;
\tcp{Find $\mathcal{A}$}   
     Compute $T(W_e^{(i)})$ for $i \in \llbracket 1; n \rrbracket$\;
     Sort $\{T(W_e^{(i)})\}_{i \in \mathcal{G}}$ 
     in increasing order and store the vector of sorting indices $\textrm{I}_e$ of size $k_b$\;
     $\delta_0=\frac{1}{2}\sum_{i \in \mathcal{B}} T(W_e^{(i)})$\; \label{line_delta_0}
     \eIf{$\delta_0 > \delta$}
     {$k_e=0$\;}
     {$k_e\!=\!\max\left\{k\in\llbracket 1; k_b\rrbracket\!: \delta_0+\frac{1}{2}\sum_{j=1}^k \!T\!\left(W_e^{\textrm{I}_e(j)}\right)\!<\delta\right\}$\;%
     \label{line_ke}
     }
       $\mathcal{A}=\{\textrm{I}_e(j) \;|\;j=1,\ldots,k_e\}$\;
    $R_s=k_e/n$\;
     \caption{Wiretap code design.} \label{algorithm_Laura}
\end{small}
\end{algorithm}

\subsection{Wiretap code design in finite blocklength}
In finite blocklength $n$, we consider a simple algorithm (Algorithm \ref{algorithm_Laura}) to choose the sets $\mathcal{A},\mathcal{R},\mathcal{B}$ in the polar coding scheme so that $P_e \leq \epsilon$ and $S(\M^k|\Z^n) \leq \delta$. 
First, the Bhattacharyya parameters for Bob's bit-channels $W_b^{(i)}$ are computed and sorted in increasing order; the set $\mathcal{G}=\mathcal{A} \cup \mathcal{R}$ is chosen as the largest possible information set 
such that the Bhattacharyya bound (\ref{Bhattacharyya_bound}) on the error probability is smaller than $\epsilon$. (Alternatively, a Monte-Carlo bound on the error probability can be used if $\epsilon>10^{-7}$). Subsequently, the TVDs for Eve's bit-channels $W_e^{(i)}$ are computed and stored in increasing order; the set $\mathcal{A}$ is chosen as the largest possible subset of $\mathcal{G}$ such that the sum of the TVDs of the bit-channels in $\mathcal{A} \cup \mathcal{B}$ (Bound 2 in Lemma \ref{lemma_bounds}) is smaller than $\delta$. 

\section{Numerical results and Discussion}
\begin{figure*}[h]
\begin{subfigure}{.5\textwidth}
  \centering
  \includegraphics[width=.95\linewidth]{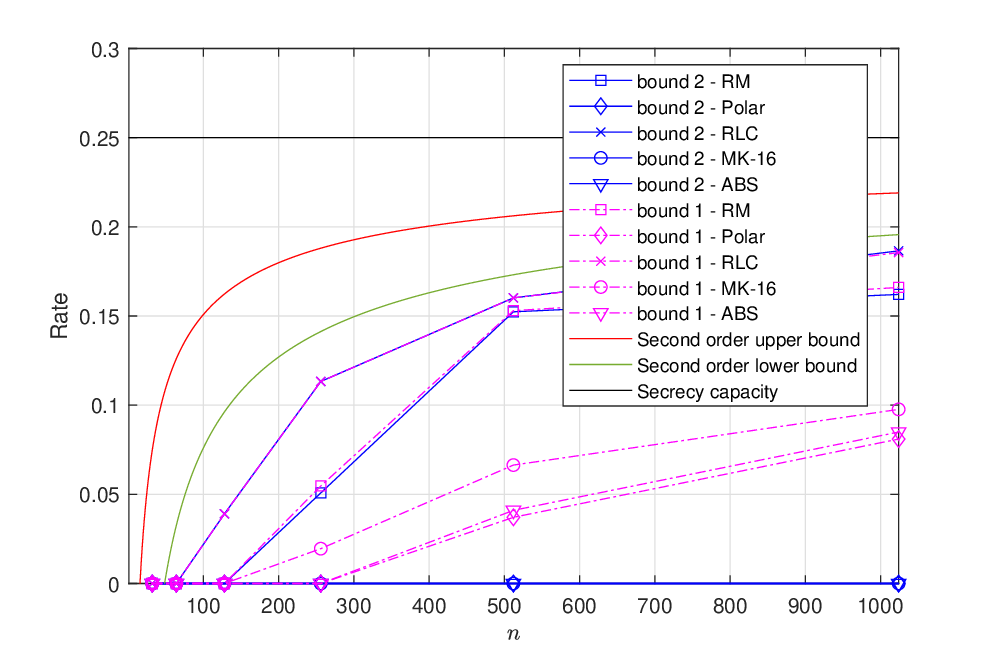}
  \caption{$p_b=0.05$ and $p_e=0.3$}
  \label{fig:res_25}
\end{subfigure}%
\begin{subfigure}{.5\textwidth}
  \centering
  \includegraphics[width=.95\linewidth]{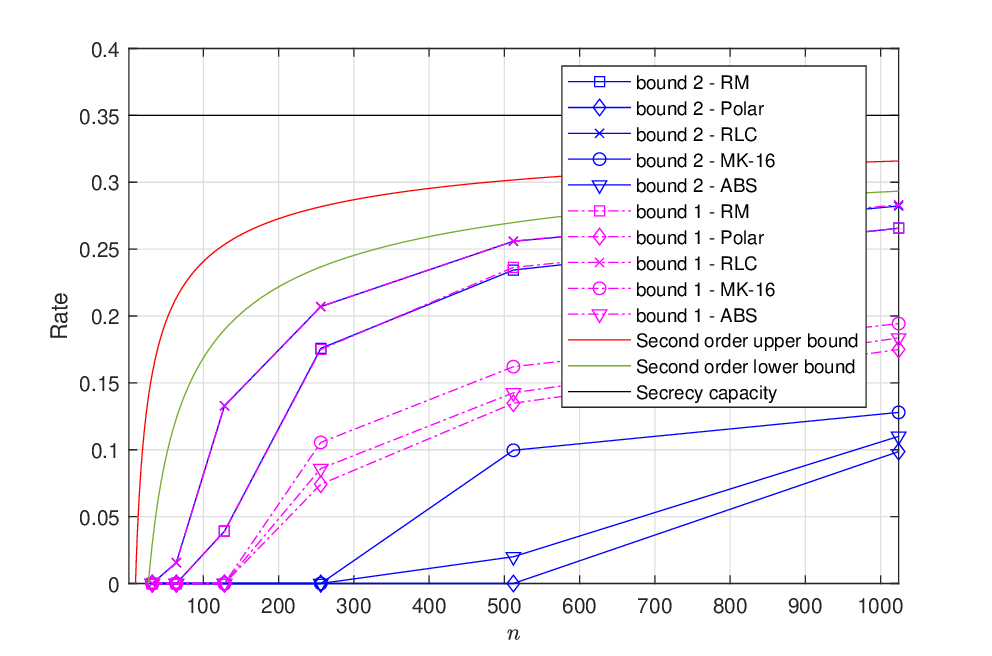}
  \caption{$p_b=0.05$ and $p_e=0.4$}
  \label{fig:res_40}
\end{subfigure}
\caption{Lower secrecy-rate bounds for polarizing codes versus second-order bounds (\ref{YSP_Theorem13_upper}) and (\ref{YSP_Theorem13_lower}) over different degraded binary erasure wiretap channels with error probability constraint $\epsilon=0.001$ and secrecy constraint $\delta=0.01$.}
\label{fig:res}
\end{figure*}
Although the bounds in Section \ref{section_wiretap} are valid for general BMS 
channels, 
the 
computation of the TVDs of the bit-channels is 
problematic in general for most codes, since the cardinality of the output alphabet grows exponentially (although for standard polar codes, they can be upper bounded using Tal and Vardy's algorithm \cite{tal2013construct}).

Therefore, in order to test different families of polarizing codes, we focus on the case where both $W_b$ and $W_e$ are BECs with erasure probabilities $p_b < p_e$.
Due to the degradedness condition, without loss of generality $W_e$ can be seen as the concatenation of $W_b$ with a BEC($p_0$), where $p_e=p_b+(1-p_b)p_0$.
For this model, $C_s=p_e-p_b$, $V_b=p_b(1-p_b)$, $V_e=p_e(1-p_e)$ and the constant in equation (\ref{V_c}) is equal to 
\begin{equation*}
    V_c = (1-p_b)p_0-(1-p_b)^2p_0^2 = (p_e-p_b)(1-(p_e-p_b)).
\end{equation*}

\subsection{Computation of Bounds 1 and 2}
Note that if $W_e$ is a BEC, then all bit-channels $W_e^{(i)}$ are also BECs with parameters $p_{e,i}$, and $T(W_e^{(i)})=1-p_{e,i}$. 
In this case, Bounds 1 and 2 can be computed by Monte-Carlo simulation, similarly to \cite{fazeli2014scaling}, for any polarizing kernel.
Let $\tilde{\U}^n$ be a permutation of $\bar{\U}^n$ such that 
$\tilde{\U}_{1}^{n-k_b}=\U_{\mathcal{B}}$, and $\tilde{\U}_j=\bar{\U}_{\mathrm{I}_e(j)}$ for $j=1,\ldots,k_b$, with $\tilde{G}_n$ being the corresponding column permutation of $G_n$. We want to estimate the erasure parameter $\tilde{p}_{e,i}$ of the channel 
$p_{\Z^n\tilde{\U}_1^{i-1}|\tilde{\U}_{i}}$, which is also a BEC for any linear code \cite{fazeli2014scaling}. 
Given an erasure pattern $\mathcal{E}$ and the set of unerased positions $\mathcal{U}=\llbracket 1; n \rrbracket \setminus \mathcal{E}$, let $\tilde{\Z}_{\mathcal{U}}$ the (permuted) output of Eve's channel restricted to $\mathcal{U}$ and $G_{\mathcal{U}}$ the restriction of $\tilde{G}_n$ to the rows corresponding to unerased positions. Then,
$$\tilde{\Z}_{\mathcal{U}}=\tilde{G}_{\mathcal{U},[1:n]}\tilde{\U}_1^{n}=\tilde{G}_{\mathcal{U},[1:i-1]}\tilde{\U}_1^{i-1}+\tilde{G}_{\mathcal{U},[i,n]}\tilde{\U}_i^{n}.$$
Supposing that the bits in $\tilde{\U}_1^{i-1}$ have been decoded, bit $\tilde{\U}_i$ is unerased if it is possible to solve for $\tilde{\U}_i$ in the linear system
$$\tilde{G}_{\mathcal{U},[i,n]}\tilde{\U}_i^{n}=\tilde{\Z}_{\mathcal{U}} \oplus \tilde{G}_{\mathcal{U},[1:i-1]}\tilde{\U}_1^{i-1},$$
or equivalently if $\rank_{\mathbb{F}_2}(\tilde{G}_{\mathcal{U},[i,n]})=\rank_{\mathbb{F}_2}(\tilde{G}_{\mathcal{U},[i+1,n]})+1$. 
We estimate $\tilde{p}_{e,i}$ by averaging over many random BEC($p_e$) erasure patterns.
The code design is described in Algorithm \ref{algorithm_Laura} for Bound 2; it can be modified for Bound 1 by replacing $T(W_e^{(i)})$ with $\tilde{p}_{e,i}$ in 
lines \ref{line_delta_0} and \ref{line_ke}. 

\subsection{Polarizing codes}
We evaluate our contruction for different families of polarizing codes, namely Polar, Multi-kernel (MK), ABS, Reed-Muller (RM) and Random Linear (RL) codes. 

RM codes have the same matrix $G_n$ of polar codes, however the columns are reordered in increasing weight order; this 
leads to 
a faster polarization \cite{abbe2020reed}, but makes the decoding unfeasible for mid-range rates. 
MK codes are constructed by combining different kernels; in particular, we use the $G_{16}$ kernel in \cite{yao2019explicit} in combination with Arikan's kernel $G_2$. 
This larger kernel permits to achieve a better polarization rate than polar codes, even if not as good as the one of RM codes, 
but entails a larger decoding complexity than polar codes \cite{bioglio2025neural}.

ABS codes \cite{li2022adjacent} represent a compromise between polar and MK codes, having an intermediate polarization rate between the two codes but decoding complexity that is still manageable. 
We adapt ABS codes to the wiretap case by swapping bits $i$ and $i+1$ in the recursive construction only when it accelerates polarization for both $W_b$ and $W_e$, i.e. when $T(W_{e}^{(i)}) \geq T(W_{e}^{(i+1)})$ and $C(W_{b}^{(i)}) \geq C(W_{b}^{(i+1)})$. 

Finally, RL codes, based on a random Bernoulli invertible matrix, upper-bound polarizing code secrecy with optimal polarization but impractical decoding \cite{fazeli2020binary}.

\subsection{Lower bounds on the achievable secrecy rates}
Figure \ref{fig:res} shows the lower bounds on the secrecy rate $R_s$ computed using Bounds 1 and 2 for the polarizing codes introduced previously. 
We consider a fixed BEC($0.05$) for Bob, and we propose two scenarios for Eve, namely a BEC($0.3$) and a BEC($0.4$); in both cases, we impose an error probability constraint $\epsilon=0.001$ and a secrecy constraint $\delta=0.01$.

Note that Bounds 1 and 2 are almost the same for RM and RL codes,
due to the fact that the virtual bit-channels are mostly ordered in reliability order, and the set $\mathcal{R}$ usually includes the last elements of $\dsb{1}{n}$. 
Moreover, RL codes achieve a better secrecy rate 
than RM codes, having a better polarization rate; their performance comes close to the lower bound  (\ref{YSP_Theorem13_lower}) for the second order secrecy rate. We conjecture that 
the lower bound (\ref{YSP_Theorem13_lower}) is tighter than the upper bound   (\ref{YSP_Theorem13_upper}).

On the other hand, ABS, Polar and MK codes exhibit a gap between Bound 1 and 2; in this case, the reordering of elements in $\mathcal{G}$ for the calculation of Bound 1 permits to slightly decrease the TVD bounds for the elements in $\mathcal{A}$, allowing for the introduction of bit-channels that were included in $\mathcal{R}$ for Bound 2. 
These results suggest that ABS codes strike a balance between achievable secrecy rate and decoding complexity, outperforming polar codes while remaining more practical than MK and RM codes.

\footnotesize
\bibliographystyle{IEEEtran}
\bibliography{IEEEabrv,wiretap}

\end{document}